\documentclass[12pt,english]{article}
\usepackage[english]{babel}
\usepackage[dvips]{graphicx}
\makeatletter
\usepackage{cite}
\usepackage{paralist}
\usepackage[cmex10]{amsmath}
\usepackage{amssymb}
\usepackage{amsthm}
\usepackage{psfrag}
\usepackage[dvips]{graphicx}
\usepackage{subfig}

\newtheorem{theorem}{Theorem}[section]

\begin{document}

\def\ss{\sigma}
\def\fpzi{\frac{\partial}{\partial z_i}}
\def\fpzj{\frac{\partial}{\partial z_j}}
\def\BM{{\bf M}}
\def\BI{{\bf I}}
\def\BA{{\bf A}}
\def\BB{{\bf B}}
\def\E{{\bf E}}
\def\BD{{\bf D}}
\def\BT{{\bf T}}
\def\BX{{\bf X}}
\def\BU{{\bf U}}
\def\Bu{{\bf u}}
\def\BV{{\bf V}}
\def\Bv{{\bf v}}
\def\BI{\hbox{I}}
\def\BJ{\hbox{J}}
\def\det{\hbox{\rm det}}
\def\tr{\hbox{\rm tr}}
\def\fpx{\frac{\partial}{\partial x}}
\def\fpy{\frac{\partial}{\partial y}}
\def\det{\hbox{\rm det}}

\title{On the Optimal Transmission Scheme to Maximize Local Capacity in Wireless Networks}

\author{
Salman Malik\footnote{INRIA Paris-Rocquencourt, France. Email: \texttt{salman.malik@inria.fr}}, 
Philippe Jacquet\footnote{INRIA Paris-Rocquencourt, France. Email: \texttt{philippe.jacquet@inria.fr}}
}
\date{}

\maketitle

\begin{abstract}

We study the optimal transmission scheme that maximizes the local capacity in two-dimensional (2D) wireless networks. Local capacity is defined as the average information rate received by a node randomly located in the network. Using analysis based on analytical and numerical methods, we show that maximum local capacity can be obtained if simultaneous emitters are positioned in a grid pattern based on equilateral triangles. We also compare this maximum local capacity with the local capacity of slotted ALOHA scheme and our results show that slotted ALOHA can achieve at least half of the maximum local capacity in wireless networks.

\end{abstract}

\section{Introduction}
\label{sec:intro}

Seminal work of Gupta \& Kumar~\cite{Gupta:Kumar} and the following works, {\it e.g.}, \cite{scaling,scaling2} quantify the capacity in wireless networks in the form of asymptotic scaling laws. However, these results may not be very useful for network protocol designers in comparing different medium access schemes that have different protocol overhead but follow the same scaling behavior. Our goal is to investigate the medium access scheme which optimizes the {\em local capacity}. Note that any such scheme may have no practical implementation but its evaluation is interesting in order to establish an upper bound on the local capacity in wireless networks. In our analysis, we will use first and second order differentiation of local capacity to prove that simultaneous emitters arranged in a grid pattern are locally optimal and, in $2D$ wireless networks, only square, hexagonal and triangular grid patterns are most optimal patterns.  

This article is organized as follows. In Section~\ref{sec:model}, we will discuss the model of our wireless network and define the local capacity. We will summarize some important related works in Section~\ref{sec:context}. The optimality of grid pattern based medium access schemes will be discussed in Section~\ref{sec:grid_optimality} and their local capacity will be analyzed in Section~\ref{sec:rx_area_2}. Section~ \ref{sec:aloha} will discuss the local capacity of simple ALOHA based scheme. In Section \ref{sec:evaluate}, we will evaluate the local capacity of grid pattern schemes and slotted ALOHA and concluding remarks can be found in Section \ref{sec:conclude}.

\section{System Model}
\label{sec:model}

We consider a wireless network where nodes are uniformly distributed over an infinite plane centered at origin $(0,0)$. We assume that time is slotted and at any given slot, simultaneous emitters in the network are distributed like a set of points, \mbox{${\cal S}=\{z_1,z_2,\ldots,z_n,\ldots\}$}, where $z_i$ is the location of emitter $i$. The distribution of set ${\cal S}$ depends on the medium access scheme employed by the nodes and we only assume that, in all slots, the set ${\cal S}$ has a homogeneous density equal to $\lambda$. 

Let $P_i$ be the transmit power of node $i$ and $\gamma_{ij}$ be the channel gain from node $i$ to node $j$ such that the received power at node $j$ is $P_i\gamma_{ij}$. Therefore, transmission from node $i$ to node $j$ is successful if the following condition is satisfied
$$
\frac{P_i\gamma_{ij}}{N_0+\sum_{k\neq i}P_k\gamma_{kj}}\geq \beta~,
$$
where $N_0$ is the background noise power and $\beta$ is the minimum signal to interference ratio (SIR) required for successfully receiving the packet. We assume that all nodes use unit nominal transmit power and we only consider large-scale pathloss characteristics, {\it i.e.}, \mbox{$\gamma_{ij}=\vert z_i-z_j\vert^{-\alpha}$}, where \mbox{$\alpha>2$} is the pathloss exponent and $\vert .\vert$ is the Euclidean norm of the vector. We also assume that $N_0$ is negligible. Therefore, the SIR of emitter $i$ at any point $z$ in the plane is given by
\begin{equation}
S_i(z)=\frac{\vert  z-z_i\vert^{-\alpha}}{\sum_{j\neq i}\vert  z-z_j\vert^{-\alpha}}~.
\label{eq:sinr}
\end{equation} 

We call the reception area of emitter $i$, the area of the plane, $A_i(\lambda,\beta,\alpha)$, where this emitter is received with SIR at least equal to $\beta$. The area, $A_i(\lambda,\beta,\alpha)$, also contains the point $z_{i}$ since here the SIR is infinite. The average size of $A_i(\lambda,\beta,\alpha)$ is $\sigma(\lambda,\beta,\alpha)$ and it is independent of the location of $z_i$.

We are interested in local capacity which is defined as the average information rate received by a node {\em randomly} located in the network. Consider a node at a random location $z$ in the plane and let $N(z,\beta,\alpha)$ be the number of reception areas it belongs to. The expected value of $N(z,\beta,\alpha)$ is given by~\cite{Jacquet:2009}
\begin{equation}
\E(N(z,\beta,\alpha))=\lambda\sigma(\lambda,\beta,\alpha)~.
\label{eq:avg_no}
\end{equation}
The average information rate received by the node, $c(z,\beta,\alpha)$, is equal to $\E(N(z,\beta,\alpha))$ multiplied by the nominal capacity. We assume unit nominal capacity and we have 
\begin{equation}
c(z,\beta,\alpha)=\E(N(z,\beta,\alpha))=\lambda\sigma(\lambda,\beta,\alpha)~. 
\label{eq:poisson_hand_over_no}
\end{equation}

\section{Motivation and Related Works}
\label{sec:context}

In related works, focus has been on the medium access schemes like ALOHA, carrier sense multiple access (CSMA) or, in some instances, node coloring as well. Some of these works are as follows. \cite{Nelson:Kleinrock} studied slotted ALOHA using a very simple geometric propagation model. Under a similar propagation model, \cite{CSMA} evaluated CSMA and compared it with slotted ALOHA in terms of throughput. \cite{Bartek} used simulations to analyze CSMA under a realistic SIR based interference model and compared it with ALOHA (both slotted and un-slotted). \cite{Weber,Weber2,Weber3} introduced the concept of {\em transmission capacity}, defined as the maximum number of successful transmissions per unit area at a specified outage probability, and evaluated ALOHA, CSMA and code division multiple access (CDMA) protocols. \cite{Haenggi} analyzed local (single-hop) throughput and capacity with slotted ALOHA, in networks with random and deterministic node placement, and TDMA, in $1D$ line-networks only. \cite{Zorzi2} determined the optimum transmission range under the assumption that interferers are distributed according to Poisson point process whereas \cite{SR-ALOHA} gave a detailed analysis on the optimal probability of transmission for ALOHA which optimizes the product of simultaneously successful transmissions per unit of space by the average range of each transmission. In contrast to these works, we investigate the most optimal medium access scheme which optimizes the local capacity in wireless networks. We will also compare this optimal scheme with slotted ALOHA. More detailed comparison with other schemes will be done in the continuation of our work.
 
\section{Optimality of Grid Pattern Schemes}
\label{sec:grid_optimality}

It can be argued that optimal local capacity in wireless networks can be achieved if simultaneous emitters are positioned in a grid pattern. However, designing a medium access scheme, which ensures that simultaneous emitters are positioned in a grid pattern, is very difficult because of the limitations introduced by wave propagation characteristics and actual node distribution. Specification of a distributed protocol that allows grid pattern transmissions is beyond the scope of this article. Note that wireless networks of grid topologies are studied in, {\it e.g.}, \cite{Liu:Haenggi,Hong:Hua}. In contrast, we assume that only the simultaneous emitters form a regular grid pattern. 

\subsection{Definitions}

In order to simplify our analysis, we define the following functions.
\begin{compactenum}

\item The density of the set ${\cal S}$ is given by the limit as
$$
\nu({\cal S})=\lim_{R\to\infty}\frac{1}{\pi R^2}\sum_i 1_{|z_i|\le R}~.
$$
\item We define a function $g(z)$ as
$$
g_i(z)=\frac{|z-z_i|^{-\alpha}}{\sum_{j}|z-z_j|^{-\alpha}}~,
$$ 
where \mbox{$\alpha>2$}. The function $g_i(z)$ is similar to the SIR function $S_i(z)$ except that the summation in the denominator factor also includes the numerator factor. In order to simplify the notations, we will remove the reference to $z$ when no ambiguity is possible. 
\item We define a function $f(g_i)$ as 
$$
f(g_i)=1_{g_i(z)\geq \beta'}~,
$$ 
for some given $\beta'$. We consider without loss of generality that the value of $\beta'$ is given by \mbox{$\beta'=\frac{\beta}{\beta+1}$}. Therefore, if transmitter $i$ is received successfully at location $z$ ({\it i.e.} with SIR at least equal to $\beta$), then \mbox{$g_i(z)\geq \beta'$} and $f(g_i)$ is equal to $1$.
\item Reception area of an arbitrary emitter $i$ is given by 
$$
\ss_i=\int f(g_i)dz^2~.
$$
Note that the integration is over the plane and the notations are simplified by taking $dxdy$ equal to $dz^2$. 
\item We define a function $h(z)$ as
$$
h(z)=\sum_{i\in {\cal S}}{f(g_i)}~,
$$ 
and it is equal to the number of emitters which can be successfully received at $z$. Note that if \mbox{$\beta>1$}, the maximum value of $h(z)$ is $1$. 
\item We also define $\E(h(z))$ by the limit as
\begin{align*}
\E(h(z))&=\lim_{R\to\infty}\frac{1}{\pi R^2}\int_{|z|\le R}h(z)dz^2~,\\
&=\lim_{R\to\infty}\frac{1}{\pi R^2}\sum_i 1_{|z_i|\le R}\ss_i=\nu({\cal S})\E(\ss_i)~,
\end{align*}
with
$$
\E(\ss_i)=\lim_{n\to\infty}\frac{1}{n}\sum_{i\le n}\ss_i~,
$$
where $n$ is the number of emitters on a disk of radius $R$ centered at $(0,0)$. As $R$ approaches infinity (because of an infinite plane), $n$ approaches infinity. Our objective is to find the spatial distribution of the set ${\cal S}$ which optimizes $\E(h(z))$. Note that $\E(h(z))$ is equivalent to $\E(N(z,\beta,\alpha))$ and local capacity given by expressions (\ref{eq:avg_no}) and (\ref{eq:poisson_hand_over_no}) respectively. 
\end{compactenum}

\subsection{First Order Differentiation}
\label{sec:grid_first_order}

Let us denote the operator of differentiation w.r.t. $z_i$ by $\nabla_i$. For \mbox{$i\neq j$}, we have 
$$
\nabla_i g_j=\alpha g_ig_j\frac{z-z_i}{|z-z_i|^2}
$$ 
and 
$$
\nabla_i g_i=\alpha(g_i^2-g_i)\frac{z-z_i}{|z-z_i|^2}~.
$$ 
Therefore
\begin{align}
\nabla_i h(z)&=\nabla_i\sum_if(g_i)=f'(g_i)\nabla_ig_i+\sum_{j\neq i}f'(g_j)\nabla_ig_j \notag \\
&=\alpha g_i\frac{z-z_i}{|z-z_i|^2}\Big(-f'(g_i)+\sum_j g_jf'(g_j)\Big)~. \notag
\end{align}
We know that \mbox{$\int h(z)dz^2=\infty$}, we nevertheless have a finite $\nabla_i \int h(z)dz^2$. In other words, the sum $\sum_{j}\nabla_i\ss_j$ converges for all $i$. For all $j$ in ${\cal S}$, \mbox{$\sum_{i}\nabla_i\ss_j=0$}. Indeed this would be the differentiation of $\ss_j$ when all points in ${\cal S}$ are translated by the same vector. Similarly, 
$$
\sum_{i}\nabla_i\int  h(z)dz^2=0~.
$$ 

\begin{theorem}
If the points in the set ${\cal S}$ are arranged in a grid pattern then 
$
\nabla_i \int h(z)dz^2=\sum_{j}\nabla_i\ss_j=0,
$ 
and grids patterns are {\em locally optimal}. 
\end{theorem}

\begin{proof}
If ${\cal S}$ is a set of points arranged in a grid pattern, then: \mbox{$\nabla_i \int h(z)dz^2=\sum_{j}\nabla_i\ss_j$} would be identical for all $i$ and, therefore, would be null since \mbox{$\sum_{i}\nabla_i\int  h(z)dz^2=0$}. 

We could erroneously conclude that,
\begin{compactitem}[-]
\item all grid sets are optimal and
\item all grid sets give the same $\E(h(z))$.
\end{compactitem}
In fact this is wrong: we could also conclude that $\E(\ss_i)$ does not vary but this will contradict that $\nu({\cal S})$ {\em must} vary. The reason of this error is that a grid set cannot be modified into another grid set with a {\em uniformly bounded transformation}, unless the two grid sets are translated by a simple vector. 
\end{proof}

\subsection{Numerical Differentiation of First and Second Order}
\label{sec:grid_second_order}

In order to prove that grid patterns are also {\em locally maximum}, we must show that
$
\nabla_i^2 \int h(z) dz^2 < 0
$.
Obviously, analytical formulation to prove this property is very challenging. However, we can develop a numerical differentiation technique to show that this is true in case of grid patterns. 

We know that \mbox{$\int h(z)dz^2=\infty$}, and we are only interested in the behavior of $\nabla_i \int h(z)dz^2$ and $\nabla_i^2 \int h(z)dz^2$. Therefore, we define a function $U$ as
\begin{equation}
U=\int_{\cal A} h(z) dz^2~,
\label{eq:capa}
\end{equation}
where the integration is over a large area, ${\cal A}$, with location of emitter $i$, $z_i$, at the center of this area. Note that the area, ${\cal A}$, is large enough so that a slight perturbation in the location of $z_i$ may have infinitesimal effect on the reception areas of the emitters near the edges of ${\cal A}$. Let $U_{i,j}$ denote the result of the integration in \eqref{eq:capa} and we denote the differentiation of $U$ w.r.t. $z_i$ as
$$
\nabla_i U=\Big(\frac{\partial}{\partial x_i}U,\frac{\partial}{\partial y_i}U\Big)=(U_x,U_y)~.
$$
We will compute $\nabla_i U$ and $\nabla_i^2 U$ with slight perturbation in $z_i$ and use the following method. 

Using central difference equations, we can write
$
U_x\approx \frac{1}{2\Delta x}(U_{i+1,j} -U_{i-1,j})
$,
and
$
U_y\approx \frac{1}{2\Delta y}(U_{i,j+1} -U_{i,j-1})
$,
where $\Delta x$ and $\Delta y$ are slight perturbations in the location of $z_i$ along $x-$axis and $y-$axis respectively. The value of $U_{i+1,j}$ is computed from \eqref{eq:capa} with perturbed position of emitter $i$ given by $z_i'=(x_i+\Delta x,y_i)$, and this also applies to other notations.

Similarly, the second order partial derivatives can be written as
$
U_{xx}\approx \frac{1}{\Delta x^2}(U_{i+1,j}-2U_{i,j}+U_{i-1,j})
$,
$
U_{yy}\approx \frac{1}{\Delta y^2}(U_{i,j+1}-2U_{i,j}+U_{i,j-1})
$,
and mixed derivative can be written as
$
U_{xy}=U_{yx}\approx \frac{1}{4\Delta x \Delta y}(U_{i+1,j+1}-U_{i+1,j-1}-U_{i-1,j+1}+U_{i-1,j-1})
$. 

In order to prove that grid patterns are locally maximum, we will analyze the eigenvalues of the $2\times2$ {\em Hessian Matrix},
$$
H(U)=\Big[\begin{array}{cc}
U_{xx} & U_{xy}\\
U_{yx} & U_{yy}\end{array}\Big]~.
$$
In other words, we can prove that the grid patterns are locally maximum if, in this case, we can show that the determinant, $|H(U)|$, is greater than zero and \mbox{$U_{xx}<0$}. 

We will introduce the analytical method to compute $U$ and its first order and second order partial derivatives in \S \ref{sec:rx_area_2} and present the results of numerical differentiation in \S \ref{sec:sub_opti}. 

We prove that the grid sets are locally optimal and maximum within sets that can be uniformly transformed between each other. In order to cope with uniform transformation and to be able to transform a grid set to another grid set, we will introduce the linear group transformation. 

\subsection{Linear Group Transformation}

Here, we assume that the points in the plane are modified according to a continuous linear transform $M(t)$ where $\BM{\bf (t)}$ is a matrix with \mbox{$\BM{\bf (0)}=\BI$}, {\it e.g.}, \mbox{$\BM{\bf (t)}=\BI+{\bf t}\BA$} where $\BA$ is a matrix. Without loss of generality, we only consider $\ss_0$, {\it i.e.}, the reception area of the emitter at $z_0$ which can be located anywhere on the plane. Under these assumptions, we have
$$
\frac{\partial}{\partial t}\ss_0=\sum_i (\BA z_i.\nabla_i\ss_0)=\tr\Big(\sum_{i} \BA^Tz_i\otimes \nabla_i\ss_0\Big)~.
$$

In other words, using the identity \mbox{$\frac{\partial \tr(\BA^T\BB)}{\partial \BA}=\BB$}, the derivative of $\ss_0$ w.r.t. matrix $\BA$ is exactly equal to \mbox{$\BD=\sum_{i} z_i\otimes \nabla_i\ss_0$}, 
such that
$
\BD=\left[
\begin{array}{cc}
D_{xx}&D_{xy}\\
D_{yx}&D_{yy}
\end{array}
\right]
$. 

We can write the following identity
$$
\tr\Big(\BA^T \frac{\partial}{\partial \BA}\ss_0\Big)=\frac{\partial}{\partial t}\ss_0(t,\BA)\Big|_{t=0}~,
$$
where $\ss_0(t,\BA)$ is the transformation of $\ss_0$ under $M(t)$, {\it i.e.}, \mbox{$\ss_0(t,\BA)=\det(\BI+\BA t)\ss_0$}. We assume that {$\BM{\bf (t)}=(1+t)\BI$} with \mbox{$\BA=\BI$}, {\it i.e.}, the linear transform is homothetic.

\begin{theorem}
$\BD$ is symmetric and $\tr(\BD)=2\ss_0$.
\end{theorem}

\begin{proof}
Under the given transform, $\ss_0(t,\BA)=\ss_0(t,\BI)=(1+t)^2\ss_0$. As a first property, we have \mbox{$\tr(\BD)=2\ss_0$}, since the derivative of $\ss_0$ w.r.t. identity matrix $\BI$ is exactly $2\sigma_0$ ({\it i.e.}, {$\tr(\BA^T\BD)=\tr(\BD)=\sigma_0'(0,\BI)=2\ss_0$}). The second property that $\BD$ is a symmetric matrix is not obvious. The easiest proof of this property is to consider the derivative of $\ss_0$ w.r.t. matrix 
\mbox{$
\BJ=\left[
\begin{array}{cc}
0&-1\\
1&0
\end{array}
\right]
$},
which is zero since $\BJ$ is the initial derivative for a rotation and reception area is invariant by rotation. 
Therefore, \mbox{$\tr(\BJ^T\BD)=D_{yx}-D_{xy}=0$}, which implies that $\BD$ is symmetric. 
\end{proof}

Note that $\BD$ can also be written in the following form
$$
\BD=\sum\limits_{i}z_i\otimes\nabla_i\ss_0=\int dz^{2}\sum\limits_{i}z_i\otimes\nabla_if(g_0)~.
$$
Let $\BT$ be defined as
$
\BT=\int dz^2\sum_i (z-z_i)\otimes\nabla_i f(g_0)~,
$
such that
$
\BD=\int \sum_{i} z\otimes \nabla_i f(g_0)dz^2 - \BT~.
$
The purpose of these definitions will become evident from theorems $3$ and $4$.

\begin{theorem}
We will show that $\int \sum_{i} z\otimes \nabla_i f(g_0)dz^2$ is equal to $\ss_0\BI$ and, therefore, $\BD=\ss_0\BI-\BT$. We will also prove that $\BT$ is symmetric. 
\end{theorem}

\begin{proof}
From the definition of $\BT$, we can see that the sum \mbox{$\sum_i (z-z_i)\otimes \nabla_i f(g_0)$} leads to a symmetric matrix since
\begin{eqnarray*}
\BT&=&\alpha\int f'(g_0)\Big(\frac{g_0^2-g_0}{|z-z_0|^2}(z-z_0)\otimes(z-z_0)+\\
&&\sum_{i\neq 0}\frac{g_0g_i}{|z-z_i|^2}(z-z_i)\otimes(z-z_i)\Big)dz^2~,
\end{eqnarray*}
and the left hand side is made of \mbox{$(z-z_i)\otimes(z-z_i)$} which are symmetric matrices. This implies that $\BT$ is also symmetric.

We can see that \mbox{$\sum_{i}\nabla_i f(g_0)=-\nabla f(g_0)$}, and using integration by parts we have
\begin{eqnarray*}
\lefteqn{\int \sum_i z\otimes \nabla_i f(g_0)dz^2=-1\times}\\
&\Biggl[
\begin{array}{cc}
\int x \fpx f(g_0)dxdy&\int x \fpy f(g_0)dxdy\\
\int y \fpx f(g_0)dxdy&\int y \fpy f(g_0)dxdy
\end{array}
\Biggr]=
\Biggl[\begin{array}{cc}
\ss_0&0\\
0&\ss_0
\end{array}\Biggr]~,
\end{eqnarray*}
which is symmetric and equal to $\ss_0\BI$. The sum/difference of symmetric matrices is also a symmetric matrix and, therefore, $\BD$ is a symmetric matrix and $\BD=\ss_0\BI-\BT$. 
\end{proof}


Now, we will only consider grid patterns and, by virtue of a grid pattern, we can have 
$$
\E(\ss_i)=\ss_0=\int f(g_0)dz^2~,
$$ 
and \mbox{$\E(h(z))=\nu({\cal S})\ss_0$}. Under homothetic transformation, $\nu({\cal S})$ and $\ss_0$ are transformed but $\nu({\cal S})\ss_0$ remains invariant.

\begin{theorem}
If the pattern of the points in set ${\cal S}$ is optimal w.r.t. linear transformation of the set, $\BD=\ss_0\BI$ and $\BT=0$.
\end{theorem}

\begin{proof}
The derivative of $\ss_0$ w.r.t. matrix $\BA$ is exactly equal to $\BD$. 
Similarly, under the same transformation
$
\frac{\partial}{\partial t}\nu({\cal S})=\frac{1}{\det(\BI+\BA t)}\nu({\cal S})
$,
and for $\BA=\BI$, it can be written as $\nu'({\cal S})(t,\BI)=\nu({\cal S})/(1+t)^2$.

In any case, the derivative of $\nu({\cal S})$ w.r.t. matrix $\BA$ is exactly equal to $-\BI\nu({\cal S})$. 
We also know that if the pattern is optimal w.r.t. linear transformation, the derivative of $\nu({\cal S})\ss_0$ w.r.t. to matrix $\BA$ shall be null. This implies that
$
\nu({\cal S})\BD-\BI\nu({\cal S})\ss_0=0~,
$
which leads to $\BD=\ss_0\BI$ and $\BT=0$.
\end{proof}

\begin{figure}[!t]
\centering
\psfrag{a}{$\sqrt{3}d$}
\psfrag{b}{$2d$}
\psfrag{c}{$d$}
\includegraphics[scale=0.65]{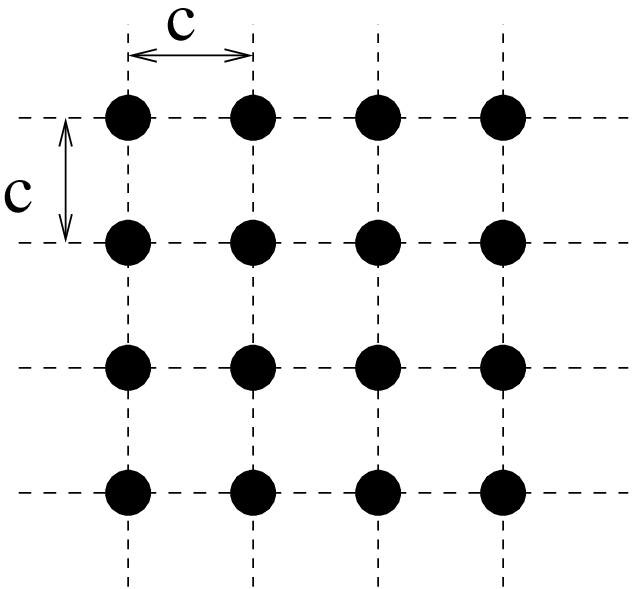}
\includegraphics[scale=0.65]{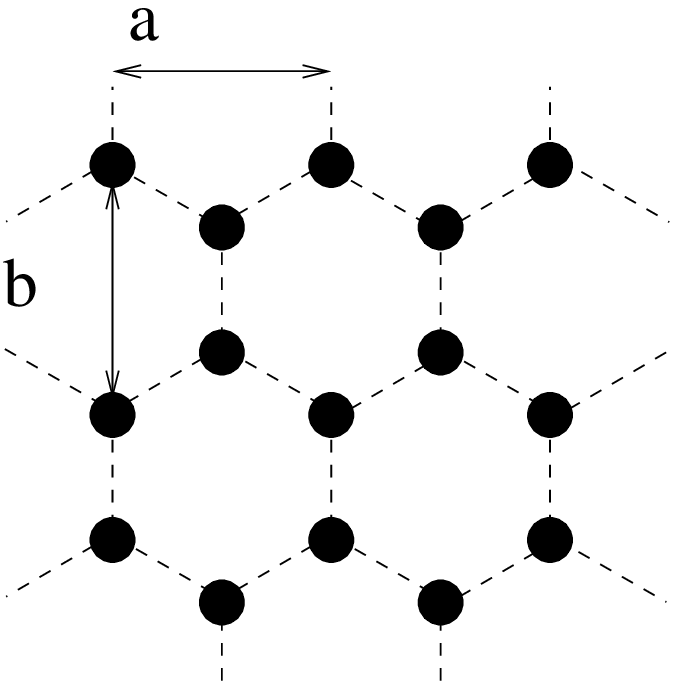}
\includegraphics[scale=0.65]{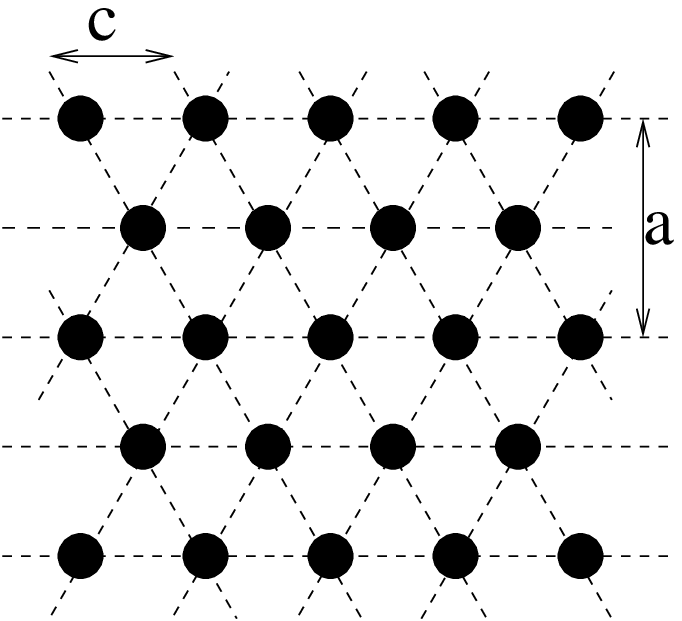}
\caption{Square, Hexagonal and Triangular grids.}
\label{fig:grid_layouts}
\end{figure}

We know that $\BT$ is symmetric and $\BT=0$. Thus, \mbox{$\tr(\BT)=0$}, {\it i.e.}, Eigen values are invariant by rotation. When a grid is optimal, we must have \mbox{$\BT=0$} and the matrix $\BT$ must be invariant w.r.t. isometric symmetries of the grid. On $2D$ plane, the grid patterns which satisfy this condition are square, hexagonal and triangular grids only. The square grid is symmetric w.r.t. any horizontal or vertical axes of the grid and, in particular, with rotation of $\pi/2$ represented by $\BJ$. Therefore, the {\em Eigen system} must be invariant by rotation of $\pi/2$. This implies that the {\em Eigen values} are the same and therefore null since \mbox{$\tr(\BT)=0$}. Same argument also applies for the hexagonal grid with the invariance for $\pi/3$ rotation and for the triangular pattern with invariance for $2\pi/3$ rotation. 

\section{Local Capacity of Grid Pattern Schemes}
\label{sec:rx_area_2}

We consider that the set ${\cal S}$ is a set of points arranged in a grid pattern and, for every slot, the grid pattern is the same {\em modulo} a translation. We have covered grid layouts of square, hexagonal and triangle which are also shown in Fig. \ref{fig:grid_layouts}. Grids are constructed from $d$ which defines the minimum distance in-between neighboring emitters and can be derived from the hop-distance parameter of a typical TDMA-based protocol. The density of grid points, $\lambda$, depends on $d$. However, the local capacity is independent of the value of $d$ or $\lambda$ as it is invariant for any homothetic transformation of the set of emitters. 

Our aim is to compute the size of the reception area, $A_i(\lambda,\beta,\alpha)$, around each emitter $i$. By consequence of the regular grid pattern, all reception areas are the same {\em modulo} a translation (and a rotation for the hexagonal pattern), and their surface area size, $\sigma(\lambda,\beta,\alpha)$, is the same.  If ${\cal C}_i(\beta,\alpha)$ is the closed curve that forms the boundary of $A_i(\lambda,\beta,\alpha)$ and $z$ is a point on ${\cal C}_i(\beta,\alpha)$, we have
\begin{equation}
\sigma(\lambda, \beta,\alpha)=\frac{1}{2}\displaystyle\int\limits_{{\cal C}_i(\beta,\alpha)}\det(z-z_{i},dz)~,
\label{eq:area_integral}
\end{equation}
where $\det(a,b)$ is the determinant of vectors $a$ and $b$ and $dz$ is the vector tangent to ${\cal C}_i(\beta,\alpha)$ at point $z$. \mbox{$\det(z-z_i,dz)$} is the cross product of vectors $(z-z_{i})$ and $dz$ and gives the area of the parallelogram formed by these two vectors. Equation (\ref{eq:area_integral}) remains true if $z_{i}$ is replaced by any interior point of $A_i(\lambda,\beta,\alpha)$. The SIR $S_{i}(z)$ of emitter $i$ at point $z$ is given by (\ref{eq:sinr}). We assume that at point $z$, $S_{i}(z)=\beta$. On point $z$ we can also define the gradient of $S_{i}(z)$, $\nabla S_{i}(z)$.
$\nabla S_i(z)$ is inward normal to the curve ${\cal C}_i(\beta,\alpha)$ and points towards $z_i$. The vector $dz$ is co-linear with $J\frac{\nabla S_{i}(z)}{|\nabla S_{i}(z)|}$ where $J$ is the anti-clockwise rotation of $3\pi/2$ (or clockwise rotation of $\pi/2$) given by
$
J=\left[\begin{array}{cc}
0 & 1\\
-1 & 0\end{array}\right]
$. 
Therefore, we can fix $dz=J\frac{\nabla S_{i}(z)}{|\nabla S_{i}(z)|}\Delta t$ and in (\ref{eq:area_integral}) 
\begin{align*}
\det(z-z_{i},dz)&=(z-z_i)\times J\frac{\nabla S_{i}(z)}{|\nabla S_{i}(z)|}\Delta t\\
&=-(z-z_{i}).\frac{\nabla S_{i}(z)}{|\nabla S_{i}(z)|}\Delta t~,
\end{align*}
where $\Delta t$ is assumed to be a small step size. The sequence of points $z(k)$ computed as 
\begin{align*}
z(0) & =z\\
z(k+1) & =z(k)+J\frac{\nabla S_{i}(z(k))}{|\nabla S_{i}(z(k))|}\Delta t~,\end{align*} gives a discretized representation of ${\cal C}_i(\beta,\alpha)$. 
Therefore, (\ref{eq:area_integral}) reduces to
\begin{equation}
\sigma(\lambda, \beta,\alpha)\approx-\frac{1}{2}\sum_{k}(z(k)-z_{i}).\frac{\nabla S_{i}(z(k))}{|\nabla S_{i}(z(k))|}\Delta t~,
\label{eq:area_integral_2}
\end{equation}
assuming that we stop the sequence $z(k)$ when it loops back on or close to the point $z$. Figure \ref{fig:rx_area_method} is the figurative representation of the computation of the reception area of an emitter. 

\begin{figure}[!t]
\centering
\psfrag{a}{$z_i$}
\psfrag{b}{$z$}
\psfrag{c}{$dz=J\frac{\nabla S_{i}(z)}{|\nabla S_{i}(z)|}\Delta t$}
\psfrag{d}{${\cal C}_i(\beta,\alpha)$}
\psfrag{e}{$A_i(\lambda,\beta,\alpha)$}
\includegraphics[scale=0.65]{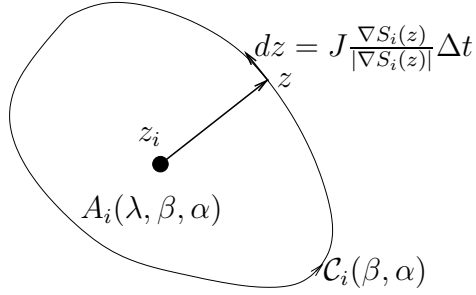}
\caption{Computation of the reception area of emitter $i$.\label{fig:snr_gradient}}
\label{fig:rx_area_method}
\end{figure}

The point, \mbox{$z(0)=z$}, can be found using Newton's method. First approximate value of $z$, required by Newton's method, can be computed assuming only one interferer nearest to the emitter $i$. The negative sign in (\ref{eq:area_integral_2}) is automatically negated by the dot product of vectors $(z(k)-z_{i})$ and $\nabla S_{i}(z(k))$.

The local capacity is given by
$$
c(z,\beta,\alpha)=\E(N(z,\beta,\alpha))=N(z,\beta,\alpha)=\lambda\sigma(\lambda,\beta,\alpha)~,
$$ 
where $\sigma(\lambda,\beta,\alpha)$ is computed using above described method. 

\section{Local Capacity of Slotted ALOHA Scheme}
\label{sec:aloha}

In slotted ALOHA scheme, the set of simultaneous emitters, in each slot, can be given by a uniform Poisson distribution of mean $\lambda$ emitters per unit square area~\cite{Jacquet:2009,SR-ALOHA,Weber2}. Therefore, using the results from~\cite{Jacquet:2009}, we can derive the analytical expression for the local capacity with slotted ALOHA scheme. 

The average size of the reception area around an arbitrary emitter is 
\begin{equation}
\sigma(\lambda,\beta,\alpha)=\frac{1}{\lambda}\frac{\sin(\frac{2}{\alpha}\pi)}{\frac{2}{\alpha}\pi}\beta^{-\frac{2}{\alpha}}~.
\label{eq:poisson_area}
\end{equation}


Therefore, the analytical expressions (\ref{eq:poisson_hand_over_no}) and (\ref{eq:poisson_area}) lead to
\begin{equation}
c(z,\beta,\alpha)=\lambda\sigma(\lambda,\beta,\alpha)=\sigma(1,\beta,\alpha)~.
\label{eq:poisson_capacity}
\end{equation}

\section{Evaluation}
\label{sec:evaluate}

In order to approach an infinite map, we perform numerical simulations in a very large network spread over $2D$ square map of $10000\times10000$ square meters. The emitters are spread over this network area in square, hexagonal or triangular pattern. For all grid patterns, \mbox{$d$} is set equal to $25$ meters although it will have no effect on the validity of our conclusions as local capacity, $c(z,\beta,\alpha)$, is independent of the value of $d$ or $\lambda$.

\subsection{Optimality of Grid Pattern Schemes}
\label{sec:sub_opti}

The area, ${\cal A}$, over which the function $U$, in \eqref{eq:capa}, is evaluated is in the center of the network area and is equal to \mbox{$2500\times2500$} square meters with $z_i$ located at its center: \mbox{$z_{i}=(x_{i},y_{i})=(0,0)$}. We set \mbox{$\beta=10.0$}, \mbox{$\alpha=4.0$} and \mbox{$\Delta x=\Delta y=0.1$}, such that the perturbation in the location of $z_i$ has an infinitesimal effect on the reception areas of the emitters on the borders of ${\cal A}$. Note that the function $U$ is numerically evaluated using the analytical method of \S \ref{sec:rx_area_2}. The first and second order partial derivatives of $U$ are computed according to \S \ref{sec:grid_second_order} and results, shown in Table \ref{tbl:comparison}, show that square, hexagonal and triangular grid patterns are locally maximum. Similar results can also be obtained for all values of $\beta$ and $\alpha$ with the same conclusions.

{\small 
\begin{table}[!t]
\begin{center}
\begin{tabular}{|c|c|c|c|c|c|c|}
\hline
     & $U_x$	& $U_y$	& $U_{xx}$	& $U_{xy}=$	& $U_{yy}$ 	& $|H(U)|$\\
     &		&		&		& $U_{yx}$	&		&		\\
   \hline

    Square   	  &  $0$ &   $0$  & $-.0052$ & $0$ & $-.0052$ & $.00002704$ \\
    \hline
    Hexagonal  &  $0$ &   $0$  & $-.0102$ & $0$ & $-.0102$ & $.00010404$ \\
    \hline
    Triangular   &  $0$ &   $0$  & $-.0041$ & $0$ & $-.0041$ & $0.0000168$ \\
       \hline
\end{tabular}
\end{center}
\caption{\footnotesize Numerical differentiation (first and second order partial derivatives) of local capacity of grid pattern schemes. $\beta=10.0$ and $\alpha=4.0$.}
\label{tbl:comparison}
\end{table}
}

\subsection{Local Capacity of grid patterns and slotted ALOHA}

In this case, we compute the size of the reception area of transmitter $i$, located in the center of the network area: \mbox{$z_{i}=(x_{i},y_{i})=(0,0)$}. The network area is large enough so that the reception area of transmitter $i$ is close to its reception area in an infinite map. $\lambda$ depends on the type of grid and it is computed from the total number of transmitters spreading over the network area. In case of slotted ALOHA, $c(z,\beta,\alpha)$ is computed from analytic expressions (\ref{eq:poisson_area}) and (\ref{eq:poisson_capacity}).

\begin{figure}[!t]
\centering
\includegraphics[scale=1.5]{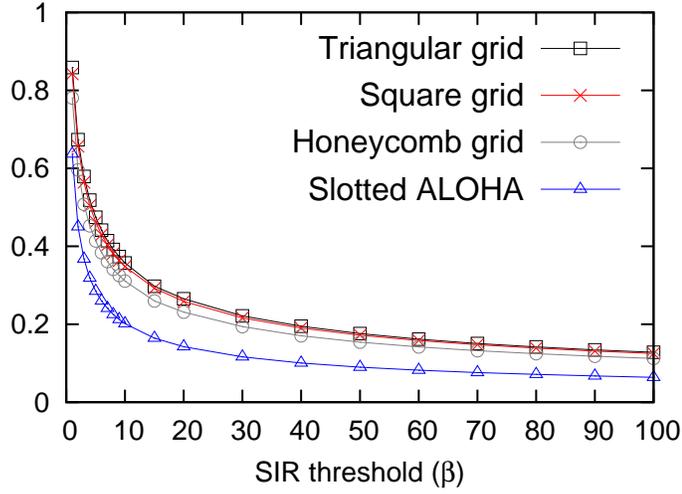}
\caption{Local capacity, $c(z,\beta,\alpha)$, when $\beta$ is varying and $\alpha$ is fixed at $4.0$.
\label{fig:comparison1}}
\end{figure}
\begin{figure}[!t]
\centering
\includegraphics[scale=1.5]{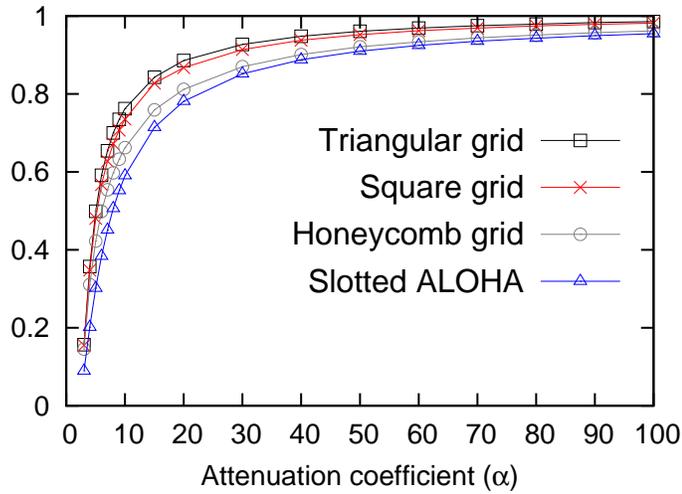}
\caption{Local capacity, $c(z,\beta,\alpha)$, when $\beta$ is fixed at $10.0$ and $\alpha$ is varying.
\label{fig:comparison2}}
\end{figure}

Figures \ref{fig:comparison1} and \ref{fig:comparison2} show $c(z,\beta,\alpha)$ of grid patterns and slotted ALOHA. As $\alpha$ approaches infinity, reception area around each transmitter turns to be a Voronoi cell with an average size equal to $1/\lambda$. Therefore, as $\alpha$ approaches infinity, $c(z,\beta,\alpha)$ approaches 1. For slotted ALOHA scheme, (\ref{eq:poisson_area}) and (\ref{eq:poisson_capacity}) arrive at the same result. For grid patterns, we computed $c(z,\beta,\alpha)$ with $\alpha$ increasing up to $100$ and from the results, we can observe that asymptotically, as $\alpha$ approaches infinity, $c(z,\beta,\alpha)$ approaching 1 is true for all protocols. We can also see that the maximum capacity in wireless networks can be obtained with triangular grid pattern and its local capacity is {\it at most} double the capacity of slotted ALOHA.

\section{Conclusions}
\label{sec:conclude}

Our analysis shows that transmission scheme, based on triangular grid pattern is locally optimal. Moreover, compared to slotted ALOHA, which does not use any significant protocol overhead, triangular grid pattern can only increase the local capacity by a factor of $2$. The conclusion of this work is that improvements above ALOHA are limited in performance and may have significantly higher protocol overheads. Note that, considering the associated protocol overheads, our results may encourage network protocol designers to concentrate on designing robust protocols based on simpler medium access schemes. In the continuation of this work, we extend this analysis to include more practical schemes like node coloring and carrier sense based schemes and see how do they compare with the schemes discussed in this article. 

\bibliographystyle{hieeetr}
\bibliography{optimal_capacity_arxiv}

\end{document}